\documentclass[letterpaper,10 pt,conference]{ieeeconf}
\IEEEoverridecommandlockouts 
\usepackage{amsmath,amssymb,amsfonts}
\usepackage{graphicx}
\usepackage{subfigure}
\usepackage{textcomp}
\usepackage{color}
\usepackage{xcolor}
\usepackage{multirow}
\usepackage{multicol,lipsum}
\usepackage{epsfig}
\usepackage{algorithm}
\usepackage{algorithmic}
\usepackage{multirow}
\usepackage{hyperref}
\usepackage{mathtools}
\usepackage{bm}
\usepackage{listings}
\newtheorem{remark}{Remark}
\newtheorem{definition}{Definition}
\newtheorem{theorem}{Theorem}
\newtheorem{problem}{Problem}

\newtheorem{lemma}{Lemma}
\newtheorem{example}{Example}

\newtheorem{assumption}{Assumption}
\usepackage{array}
\usepackage[short]{optidef}
\usepackage{mdframed}
\mdfdefinestyle{mystyle}{
    backgroundcolor=white!10}

\newcommand{\oomit}[1]{}

\newcolumntype{M}[1]{>{\centering\arraybackslash}m{#1}}
\newcolumntype{N}{@{}m{0pt}@{}}

\begin{document}

\title{Reach-avoid Analysis for Sampled-data Systems with Measurement Uncertainties}

\author{Taoran Wu$^{1,2}$, Dejin Ren$^{1,2}$, Shuyuan Zhang$^{3}$, Lei Wang$^{3}$, and Bai Xue$^{1,2}$
\thanks{1. State Key Lab. of Computer Science, Institute of Software, CAS, Beijing, China. Email: \{rendj,wutr,xuebai\}@ios.ac.cn
}
\thanks{2. University of Chinese Academy of Sciences, Beijing, China}
\thanks{3.  Beihang University, Beijing, China. Email: \{zhshuyuan,lwang\}@buaa.edu.cn
}
}
        
%

\maketitle
\thispagestyle{empty}
\pagestyle{empty}

\begin{abstract}
Digital control has become increasingly prevalent in modern systems, making continuous-time plants controlled by discrete-time (digital) controllers ubiquitous and crucial across industries, including aerospace, automotive, and manufacturing. This paper focuses on investigating the reach-avoid problem in such systems, where the objective is to reach a goal set while avoiding unsafe states, especially in the presence of state measurement uncertainties. We propose an approach that builds upon the concept of exponential control guidance-barrier functions, originally used for synthesizing continuous-time feedback controllers. We introduce a sufficient condition that, if met by a given continuous-time feedback controller, ensures the safe guidance of the system into the goal set in its sampled-data implementation, despite state measurement uncertainties. The event of reaching the goal set is determined based on state measurements obtained at the sampling time instants. Numerical examples are provided to demonstrate the validity of our theoretical developments, showcasing successful implementation in solving the reach-avoid problem in sampled-data systems with state measurement uncertainties.

\end{abstract}


\section{Introduction}
\label{sec:intro}
The problem of guiding a dynamical system towards a desired set while avoiding unsafe states is known as reach-avoid analysis. This problem is crucial in domains such as robot motion planning as well as safety and performance critical control. To address this problem, various approaches have been proposed to synthesize reliable continuous-time feedback controllers \cite{margellos2011hamilton, majumdar2014convex, fan2018controller, xue2023reach_ode}. However, despite the theoretical soundness of continuous-time controller designs, implementing them on real systems can be challenging due to the continuous updates required. As a result, in practical applications, the continuous-time controller is often discretized and implemented using discrete-time sampling techniques \cite{ackermann2012sampled}, leading to the creation of sampled-data control systems. The analysis of these control systems is complicated by the presence of both continuous and discontinuous (discrete) components in the control loop.

Sampled-data control systems involve the periodic or aperiodic measurement of the state evolution of a continuous-time plant, with a constant control signal applied until the next sample time \cite{chen2012optimal}. These systems offer various advantages over continuous-time control systems \cite{buso2015digital}, such as improved precision and fault tolerance, making them widely adopted in industries ranging from robotics and process control to automotive systems and industrial automation. However, they have a limitation in that they do not consider the states that the plant transitions through between sample times, as the system state is only observable at each sampling time. This limitation becomes particularly significant for systems with fast dynamics or unstable behavior. Additionally, state measurements in sampled-data control systems are prone to imperfections, which can lead to inaccurate state estimation. Such inaccuracies have the potential to degrade system performance and even violate reach-avoid properties if not adequately addressed in the design of the continuous-time controller. Thus, it is crucial to account for the impact of state estimation errors and the system behavior between sampling times when developing the continuous-time controller, especially for safety and performance critical systems.

This paper explores the reach-avoid problem associated with sampled-data control systems, considering the presence of uncertainties in state measurements. The analysis focuses specifically on periodic measurements. Leveraging exponential control guidance-barrier functions \cite{xue2023reach}, our aim is to establish a sufficient condition that ensures the safety of a given Lipschitz feedback control in its sampled-data implementation, with the objective of driving the system towards a desired goal set. The attainment of goal reachability relies on state measurements taken at sampling instants. To achieve this, we scrutinize the dynamic disparities between the closed-loop system operating with the feedback controller and its corresponding sampled-data control system, while accounting for uncertainties in state measurements. Lastly, we present illustrative examples to showcase the theoretical advancements of our proposed method.

This paper makes two main contributions. Firstly, it investigates the dynamics discrepancy between a control system with a Lipschitz feedback controller and its corresponding sampled-data control system which incorporates state measurement uncertainties. Secondly, building upon this investigation and utilizing exponential control guidance-barrier functions, it proposes a sufficient condition to ensure the satisfaction of reach-avoid specifications for the sampled-data control system. By addressing these contributions, the paper sheds light on the controller design of sampled-data control systems using the emulation-based approach (i.e., approximating an available continuous-time controllers) \cite{chen2012optimal} and presents practical solutions for fulfilling reach-avoid specifications.




\subsection{Related Work}
This subsection reviews works which are closely related to the present one.  

In the past two decades, extensive research has been conducted on sampled-data systems \cite{hetel2017recent,verdier2022formal}. However, the synthesis of safety and performance critical controllers for these systems is still relatively new and thus offers an open area for further investigation. In recent years, control barrier functions (CBFs) have emerged as a popular tool for constructing safety-critical controllers that provide rigorous safety guarantees for nonlinear systems. Initially developed for continuous-time systems, CBFs have also been adapted for sampled-data systems \cite{cortez2019control,gurriet2019realizable,singletary2020control,breeden2021control,niu2021safety}. These adaptations incorporate a margin term into the standard CBF derivative condition to account for potential changes in the dynamics and CBF during the inter-sample period.

 In addition to safety, the ability to reach desired goal sets is also crucial in many practical applications  \cite{fisac2015reach,fan2018controller}. 
Therefore, controllers for such systems should not only prioritize safety, but also be aware of reachability. This is especially crucial for safety-critical and performance-critical systems. These controllers need to adhere to control input constraints and safety requirements while also facilitating the system's achievement of specific states. Existing literature has explored the combination of control barrier functions and control Lyapunov functions to achieve both safety and stability \cite{ames2016control}. Recently, a novel approach called control guidance-barrier functions has been introduced to synthesize continuous-time feedback controllers that enable safe reaching of goal sets \cite{xue2023reach_ode,xue2023reach}. Expanding upon the concept of exponential control guidance-barrier functions introduced in \cite{xue2023reach}, this research aims to extend the analysis to encompass reach-avoid scenarios in sampled-data control systems.

Some basic notions are used throughout this paper: $\mathbb{R}_{\geq 0}$, $\mathbb{R}_{>0}$  and  $\mathbb{R}$ stand for the set of nonnegative reals, positive reals and real numbers, respectively. $\mathbb{N}_{\geq 0}$ denotes the set of non-negative integers. For a set $\mathcal{A}$, $\mathcal{A}^c$, $\overline{\mathcal{A}}$ and $\partial \mathcal{A}$ denote the complement, the closure and the boundary of the set $\mathcal{A}$, respectively. $\wedge$ denotes the logical operation of conjunction. 

\section{Preliminaries}
\label{sec:pre}
In this section we review exponential control guidance-barrier functions, initially proposed for synthesizing continuous-time feedback controllers to enforce reach-avoid objectives in \cite{xue2023reach}.

Consider a control affine system
\begin{equation}
\label{system}
    \dot{\bm{x}}=\bm{F}(\bm{x},\bm{u}):=\bm{f}(\bm{x})+\bm{g}(\bm{x})\bm{u},
\end{equation}
where $\bm{x}\in \mathbb{R}^n$, $\bm{u}\in \mathcal{U}\subseteq \mathbb{R}^m$ with $\mathcal{U}=\{\bm{u}\in \mathbb{R}^m \mid \|\bm{u}\|\leq \overline{u}\}$, and $\bm{f}:\mathbb{R}^n\rightarrow \mathbb{R}^n$ and $\bm{g}: \mathbb{R}^n\rightarrow \mathbb{R}^{n\times m}$ are locally Lipschitz continuous. Given a locally Lispchitz feedback control law $\bm{k}:\mathbb{R}^n\rightarrow \mathcal{U}$ and an initial state $\bm{x}_0\in \mathbb{R}^n$, let $\bm{x}(\cdot;\bm{x}_0): \mathcal{I}\rightarrow \mathbb{R}^n$ be the resulting solution of the closed-loop system \eqref{system} under the control signal $\bm{u}(t):=\bm{k}(\bm{x}(t))$ defined on some maximal interval of existence $\mathcal{I}\subseteq \mathbb{R}_{\geq 0}$.   

Let $h_{\mathcal{C}}(\bm{x}): \mathcal{D}\rightarrow \mathbb{R}$ be a continuously differentiable and bounded function, which defines a safe set $\mathcal{C}$ that system \eqref{system} should satisfy in the following way: 
\begin{equation}
    \mathcal{C}=\{\bm{x}\in \mathcal{D}\mid h_{\mathcal{C}}(\bm{x})>0\},
\end{equation}
where $\overline{\mathcal{C}}\subseteq \mathcal{D}\subseteq \mathbb{R}^n$ and $\partial \mathcal{C}=\{\bm{x}\in \mathcal{D}\mid h_{\mathcal{C}}(\bm{x})=0\}$. In addition, let a goal set $\mathcal{G}\subseteq \mathcal{C}$ be defined by:
 \begin{equation}
     \mathcal{G}=\{\bm{x}\in \mathcal{C}\mid h_{\mathcal{G}}(\bm{x})>0\},
 \end{equation}
 where $h_{\mathcal{G}}(\bm{x}): \mathcal{D}\rightarrow \mathbb{R}$ is a Lispchitz continuous function. 

 Given a locally Lipschitz continuous control law $\bm{k}(\bm{x}):\mathcal{D}\rightarrow \mathcal{U}$, the reach-avoid property with respect to the safe set $\mathcal{C}$ and goal set $\mathcal{G}$ is a property that system \eqref{system} starting from any state in $\mathcal{C}$ can hit the goal set $\mathcal{G}$ eventually while staying inside the set $\mathcal{C}$ before the first goal hitting time. It is formally stated in Definition \ref{r_a_p}. 

 \begin{definition}
 \label{r_a_p}
     Consider system \eqref{system} with $\bm{x}(0)=\bm{x}_0$ under a given locally Lipschitz continuous control policy $\bm{k}(\bm{x}):\mathcal{D}\rightarrow \mathbb{R}^m$ and the maximal interval of existence $\mathcal{I}\subseteq \mathbb{R}_{\geq 0}$ of the solution $\bm{x}(t;\bm{x}_0)$. The system \eqref{system}  satisfies the reach-avoid property with respect to the safe set $\mathcal{C}$ and goal set $\mathcal{G}$ if there exists $\tau\in \mathcal{I}$ such that $\bm{x}(\tau;\bm{x}_0)\in \mathcal{G}$ and $\bm{x}(t;\bm{x}_0)\in \mathcal{C}$ for all $t\in [0,\tau]$, i.e.,
     \[\exists \tau\in \mathbb{R}_{\geq 0}. \bm{x}(\tau;\bm{x}_0)\in \mathcal{G} \wedge \forall t\in [0,\tau]. \bm{x}(\tau;\bm{x}_0)\in \mathcal{C}.\]
 \end{definition}

If $h_{\mathcal{C}}(\bm{x}): \mathcal{D}\rightarrow \mathbb{R}$ is an exponential control guidance-barrier function for system \eqref{system} with the controller $\bm{k}(\bm{x})$, then the closed-loop system \eqref{system} satisfies the reach-avoid property with respect to  the safe set $\mathcal{C}$ and goal set $\mathcal{G}$ \cite{xue2023reach}. 

 \begin{definition}
 \label{egbf_de}
  Consider system \eqref{system} with a locally Lipschitz control law $\bm{k}(\bm{x}):\mathcal{D}\rightarrow \mathbb{R}^m$, 
  $h_{\mathcal{C}}(\bm{x}): \mathcal{D}\rightarrow \mathbb{R}$ is called an exponential control guidance-barrier function if there exists a positive value $\lambda\in \mathbb{R}_{>0}$ satisfying
  \begin{equation}
  \label{egbfc}
      \mathcal{L}_{\bm{f}}h_{\mathcal{C}}(\bm{x})+\mathcal{L}_{\bm{g}}h_{\mathcal{C}}(\bm{x})\bm{k}(\bm{x})-\lambda h_{\mathcal{C}}(\bm{x})\geq 0, \forall \bm{x}\in \mathcal{C}\setminus \mathcal{G},
  \end{equation}
  where $\mathcal{L}_{\bm{f}}h_{\mathcal{C}}(\bm{x})=\frac{\partial h_{\mathcal{C}}(\bm{x})}{\partial \bm{x}} \bm{f}(\bm{x})$ and $\mathcal{L}_{\bm{g}}h_{\mathcal{C}}(\bm{x})=\frac{\partial h_{\mathcal{C}}(\bm{x})}{\partial \bm{x}} \bm{g}(\bm{x})$. 
 \end{definition}


 \begin{theorem}
 \label{r_a_c}
     If  $h_{\mathcal{C}}(\bm{x}): \mathcal{D}\rightarrow \mathbb{R}$ is an exponential control guidance-barrier function with respect to the locally Lipschitz control law $\bm{k}(\bm{x})$, then system \eqref{system} with the control law  $\bm{k}(\bm{x})$ satisfies the reach-avoid property with respect to the safe set $\mathcal{C}$ and goal set $\mathcal{G}$.
 \end{theorem}

\section{Problem Formulation}
\label{sec:pf}

In this section we formulate the reach-avoid problem of interest for sampled-data systems subject to state measurement errors.  

For a sampled-data system, the sampling instants are described by a sequence of strictly increasing positive real numbers $\{t_i\}, i\in \mathbb{N}_{\geq 0}$, where $t_0=0, t_{i+1}-t_{i}>0$ and $\lim_{i\rightarrow \infty}t_i=\infty$. The system's states are observable only at these sampling instants. Define the sampling interval between $t_i$ and $t_{i+1}$ as 
\[\Delta_i=t_{i+1}-t_{i}.\]
The sampling mechanism is called a periodic if $\Delta_i$ are the same for all $i$, and an aperiodic sampling otherwise. In this work we consider the former with $\Delta_i=\Delta$ for $i\in \mathbb{N}$, i.e., the sampling mechanism is periodic. The control input for a sampled-data system is a piecewise constant signal with respect to $t_i$, i.e., 
\[\bm{u}(t)=\bm{u}_i, \forall t\in [t_i,t_{i+1}).\]
Consequently, when the state of system \eqref{system} is $\bm{x}_i$ at $t=t_i$, the dynamics of system \eqref{system} in the sampled-data form over $[t_i,t_{i+1})$ are governed by  
\begin{equation}
\label{sc_sys}
    \dot{\bm{x}}=\bm{f}(\bm{x})+\bm{g}(\bm{x})\bm{u}_i, \bm{x}(t_i)=\bm{x}_i, \forall t\in [t_i,t_{i+1}).
\end{equation}
We denote the solution to \eqref{sc_sys} as $\bm{x}_{s}(\tau-t_i;\bm{x}_i)$ for $\tau\in [t_i,t_{i+1})$, where $\bm{x}_{s}(0;\bm{x}_i)=\bm{x}_i$. 

Given system \eqref{system} with a locally Lipschitz continuous controller $\bm{k}(\bm{x}): \mathcal{D}\rightarrow \mathcal{U}$, we in this paper consider its corresponding sampled-data system taking the control input $\bm{u}(t)$ in a
zero-order hold (ZOH) manner,  
\begin{equation}
\label{sc}
    \bm{u}(t)=\bm{k}(\bm{x}_i), \forall t\in [t_i,t_{i+1}).
\end{equation}
where $\bm{x}_i=\bm{x}(t_i)$ is the system's state at time $t=t_i$. Also, due to the measurement error, the evolution of the resulting sampled-data system is actually governed by 
\begin{equation}
\label{sc_system_k}
    \dot{\bm{x}}=\bm{f}(\bm{x})+\bm{g}(\bm{x})\bm{k}(\hat{\bm{x}}_i), \bm{x}(t_i)=\bm{x}_i, \forall t\in [t_i,t_{i+1}),
\end{equation}
where $\hat{\bm{x}}_i$ is the measured state at time $t=t_i$,  $\|\hat{\bm{x}}_i-\bm{x}_i\|\leq \epsilon$ with $\bm{x}_i$ being the exact state at time $t=t_i$ which is unknown in practice, and $\epsilon$ is the bound of the measurement error. 

We denote the solution to \eqref{sc_system_k} as $\bm{x}_{s,k}(\tau-t_i;\bm{x}_i)$ for $\tau\in [t_i,t_{i+1})$, where $\bm{x}_{s,k}(0;\bm{x}_i)=\bm{x}_i$. 


According to Theorem \ref{r_a_c}, a Lipschitz continuous-time controller $\bm{k}(\cdot): \mathcal{D}\rightarrow \mathcal{U}$ that satisfies constraint \eqref{egbfc} can safely guide the system \eqref{system} from any state in $\mathcal{C}$ to the goal set $\mathcal{G}$. However, this reach-avoid property may not hold for the sampled-data system \eqref{sc_system_k} due to the discrepancy between this continuous-time controller and its sampled-time implementation. Therefore, it is crucial to improve the performance of the controller $\bm{k}(\cdot): \mathcal{D}\rightarrow \mathcal{U}$ to ensure that the sampled-data system \eqref{sc_system_k} satisfies this reach-avoid property. Furthermore, since the states of the sampled-data system \eqref{sc_system_k} are only observable at sampling time instants, it is only possible to determine successful entry into the goal set $\mathcal{G}$ at these specific time points. Therefore, it is crucial to improve the controller's performance in a manner that enables the determination of goal reach for the sampled-data system \eqref{sc_system_k} based on measured states at sampling instants.

Based on the considerations outlined above, we finalize our reach-avoid problem of interest as Problem \ref{reach_avoid_set_sampled}.

\begin{problem}[Reach-avoid Satisfaction of System \eqref{sc_system_k}]
\label{reach_avoid_set_sampled}
We propose a modified condition, derived from the one \eqref{egbfc}. If a locally Lipschitz continuous controller $\bm{k}(\cdot): \mathcal{D}\rightarrow \mathcal{U}$ satisfies this condition, the sampled-data system \eqref{sc_system_k} will, starting from any state in $\mathcal{C}$, reach the goal set $\mathcal{G}$ at a sampling time instant $t=t_i$. Additionally, it will remain within the safe set $\mathcal{C}$ before $t_i$, and the confirmation of goal reach will be based on the measured state $\hat{\bm{x}}_i$.
\end{problem}

In order to solve Problem \ref{reach_avoid_set_sampled}, we need the following assumptions. 
\begin{assumption}
 \label{assumption}
 \begin{enumerate}
     \item  $\sup_{\bm{x}\in \mathcal{D},\bm{u}\in \mathcal{U}} \|\bm{f}(\bm{x})+\bm{g}(\bm{x})\bm{u}\|\leq \alpha$;
     \item $\|\bm{k}(\bm{x})-\bm{k}(\bm{y})\|\leq \beta\|\bm{x}-\bm{y}\|, \forall \bm{x},\bm{y}\in \mathcal{D}$;
     \item $\sup_{\bm{x}\in \mathcal{D}}\|\frac{\partial h_{\mathcal{C}}(\bm{x})}{\partial \bm{x}}\bm{g}(\bm{x})\|\leq \gamma$;
     \item $\mathcal{C}\bigoplus \mathcal{B}_{\epsilon}\subseteq \mathcal{D}$, where $\mathcal{B}_{\epsilon}=\{\bm{y}\in \mathbb{R}^n\mid \|\bm{y}\|\leq \epsilon\}$, $\bigoplus$ denotes the Minkowski sum and $\epsilon$ is the upper bound of the state measurement error.
     \item $\|h_{\mathcal{G}}(\bm{x})-h_{\mathcal{G}}(\bm{y})\|\leq \xi\|\bm{x}-\bm{y}\|, \forall \bm{x},\bm{y}\in \mathcal{D}$.
 \end{enumerate}    
 \end{assumption}
\section{Reach-avoid Analysis}
\label{sec:prob1}


In this section, we will address Problem \ref{reach_avoid_set_sampled}. We aim to establish a sufficient condition for system \eqref{sc_system_k} to satisfy the reach-avoid property with respect to the safe set $\mathcal{C}$ and goal set $\mathcal{G}$. To achieve this, we first analyze the dynamic discrepancies between systems \eqref{system} and \eqref{sc_system_k}, which arise due to the sampled-data implementation of the locally Lipschitz controller $\bm{k}(\cdot): \mathcal{D}\rightarrow \mathcal{U}$. By conducting this analysis, we can assess the reach-avoid property of system \eqref{sc_system_k} by examining system \eqref{system} under certain perturbation inputs.

This section is structured as follows: In Subsection \ref{sec:regbf}, we introduce robust exponential control guidance-barrier functions, which are relevant to the reach-avoid analysis for a perturbed system. The perturbed system is obtained by adding perturbation inputs into system \eqref{system} subject to the locally Lipschitz controller $\bm{k}(\cdot): \mathcal{D}\rightarrow \mathcal{U}$. In Subsection \ref{sec:ddc}, we analyze the dynamic discrepancies between system \eqref{system} with the locally Lipschitz controller $\bm{k}(\cdot): \mathcal{D}\rightarrow \mathcal{U}$ and \eqref{sc_system_k}. In Subsection \ref{sec:scrass}, we establish the sufficient condition for system \eqref{sc_system_k} to satisfy the reach-avoid property, solving Problem \ref{reach_avoid_set_sampled}.

\subsection{Robust Exponential Control Guidance-barrier Functions}
\label{sec:regbf}
In this subsection, we introduce robust exponential control guidance-barrier functions.

The perturbed system, which is obtained by adding perturbation inputs into system \eqref{system} with the locally Lipshcitz controller $\bm{k}(\cdot): \mathcal{D}\rightarrow \mathcal{U}$, is of the following form
\begin{equation}
\label{p_system}
    \dot{\bm{x}}(t)=\overline{\bm{f}}(\bm{x}(t))+\bm{g}(\bm{x}(t))\bm{d}(t), \bm{x}(0)=\bm{x}_0,
\end{equation}
where $\overline{\bm{f}}(\bm{x})=\bm{f}(\bm{x})+\bm{g}(\bm{x})\bm{k}(\bm{x})$ and $\bm{d}(\cdot): \mathbb{R}\rightarrow \mathbb{R}^m$ is the perturbation input satisfying $\|\bm{d}\|\leq \overline{d}$. 

 Robust exponential control guidance-barrier functions are an exponential control guidance-barrier function that takes perturbation inputs into account.

\begin{definition}
    Consider system \eqref{p_system} with a locally Lipschitz control law $\bm{k}(\bm{x}):\mathcal{D}\rightarrow \mathbb{R}^m$, $h_{\mathcal{C}}(\bm{x}): \mathcal{D}\rightarrow \mathbb{R}$ is called a robust exponential guidance-barrier function if there exists a positive value $\lambda\in \mathbb{R}_{>0}$ satisfying
  \begin{equation}
  \label{regbfc}
      \mathcal{L}_{\bm{f}}h_{\mathcal{C}}(\bm{x})+\mathcal{L}_{\bm{g}}h_{\mathcal{C}}(\bm{x})\bm{k}(\bm{x})-\lambda h_{\mathcal{C}}(\bm{x})\geq \gamma \overline{d}, \forall \bm{x}\in \mathcal{C}\setminus \mathcal{G}.
  \end{equation}
\end{definition}
 
  Robust exponential control guidance-barrier functions function guarantee the fulfillment of the reach-avoid property for  system \eqref{p_system} with respect to the safe set $\mathcal{C}$ and goal set $\mathcal{G}$, regardless of any perturbation inputs.
 \begin{theorem}
 \label{perturbed_reach_avoid}
 If $h_{\mathcal{C}}(\bm{x}):\mathcal{D}\rightarrow \mathbb{R}$ is a robust exponential control guidance-barrier function with respect to the locally Lipschitz control law $\bm{k}(\bm{x}):\mathcal{D}\rightarrow \mathcal{U}$, then system \eqref{p_system} with the control law $\bm{k}(\bm{x})$ satisfies the reach-avoid property with respect to the safe set $\mathcal{C}$ and goal set $\mathcal{G}$, regardless of perturbation inputs. 
 \end{theorem}
 \begin{proof}
  Since $\mathcal{L}_{\bm{f}}h_{\mathcal{C}}(\bm{x})+\mathcal{L}_{\bm{g}}h_{\mathcal{C}}(\bm{x})\bm{k}(\bm{x})-\lambda h_{\mathcal{C}}(\bm{x})\geq \gamma \overline{d}, \forall \bm{x}\in \mathcal{C}\setminus \mathcal{G}$, we have that  for all $\bm{d}(\bm{x})$ satisfying $\|\bm{d}\| \leq \overline{d}$, \[
\begin{split}  
  \mathcal{L}_{\bm{f}}h_{\mathcal{C}}(\bm{x})&+\mathcal{L}_{\bm{g}}h_{\mathcal{C}}(\bm{x})\bm{k}(\bm{x})+\mathcal{L}_{\bm{g}}h_{\mathcal{C}}(\bm{x})\bm{d}(\bm{x})\\
  &-\lambda h_{\mathcal{C}}(\bm{x})\geq \gamma \overline{d}+\mathcal{L}_{\bm{g}}h_{\mathcal{C}}(\bm{x})\bm{d}, \forall \bm{x}\in \mathcal{C}\setminus \mathcal{G}.
  \end{split}
  \]
  Consequently, 
  \[
  \begin{split}
  \mathcal{L}_{\bm{f}}h_{\mathcal{C}}(\bm{x})&+\mathcal{L}_{\bm{g}}h_{\mathcal{C}}(\bm{x})\bm{k}(\bm{x})\\
  &+\mathcal{L}_{\bm{g}}h_{\mathcal{C}}(\bm{x})\bm{d}(\bm{x})-\lambda h_{\mathcal{C}}(\bm{x})\geq 0, \forall \bm{x}\in \mathcal{C}\setminus \mathcal{G}
  \end{split}
  \]
  holds for all $\bm{d}$ satisfying $\|\bm{d}\| \leq \overline{d}$. 
 According to Theorem \ref{r_a_c}, we have the conclusion.
 \end{proof}

\subsection{Dynamic Discrepancy Characterization}
\label{sec:ddc}
In this subsection, we characterize the dynamic discrepancy between system \eqref{system} with the controller $\bm{k}(\bm{x}_{s,k}(t-t_i;\bm{x}_i)):[t_i,t_{i+1})\rightarrow \mathcal{U}$ and \eqref{sc_system_k} over the time interval $[t_i,t_{i+1})$. 

We first estimate the discrepancy between the controller $\bm{k}(\bm{x}_{s,k}(t-t_i;\bm{x}_i)): [t_i,t_{i+1})\rightarrow \mathcal{U}$ and $\bm{k}(\hat{\bm{x}}_i)$ over the time interval $[t_i,t_{i+1})$, where $\bm{x}_{s,k}(t-t_i;\bm{x}_i):[t_i,t_{i+1})$ is the solution to system \eqref{sc_system_k} with $\bm{x}(0;\bm{x}_i)=\bm{x}_i$. 
  \begin{lemma}
  \label{perturbation_bound}
 Let $\bm{x}_{s,k}(t-t_i;\bm{x}_i)\in \mathcal{C}$ for $t\in [t_i,t'_i]$ with $t'_i<t_{i+1}$, then \[\|\bm{k}(\bm{x}_{s,k}(t-t_i;\bm{x}_i))-\bm{k}(\hat{\bm{x}}_i)\|\leq \min\{\beta \epsilon+\beta \alpha \Delta,2\overline{u}\}\] for $t\in [t_i,t'_{i}]$, where $\hat{\bm{x}}_i$ is the measured value of the exact state 
 $\bm{x}_i$ at time $t=t_i$ and $\|\hat{\bm{x}}_i-\bm{x}_i\|\leq \epsilon$. 
  \end{lemma}
  \begin{proof}
  According to Assumption \ref{assumption}, we have that 
  $\hat{\bm{x}}_i\in \mathcal{D}$. Therefore, for $t\in [t_i,t'_i]$,
  \begin{equation*}
  \begin{split}
 & \|\bm{k}(\bm{x}_{s,k}(t-t_i;\bm{x}_i))-\bm{k}(\hat{\bm{x}}_i)\|\\
 &\leq \|\bm{k}(\bm{x}_{s,k}(t-t_i;\bm{x}_i))-\bm{k}(\bm{x}_i)\|+\|\bm{k}(\bm{x}_i)-\bm{k}(\hat{\bm{x}}_i)\|\\
 &\leq \beta \|\bm{x}_{s,k}(t-t_i;\bm{x}_i)-\bm{x}_i\|+\beta \|\bm{x}_i-\hat{\bm{x}}_i\|\\
 &\leq \beta \int_{0}^{t-t_i} \|\bm{f}(\bm{x}_{s,k}(\tau;\bm{x}_i))+\bm{g}(\bm{x}_{s,k}(\tau;\bm{x}_i))\bm{k}(\hat{\bm{x}}_i)\|d\tau +\beta \epsilon\\
 &\leq \beta \alpha \Delta+\beta \epsilon.
  \end{split}
  \end{equation*}
  In addition, since $\bm{k}(\bm{x}): \mathcal{D}\rightarrow \mathcal{U}$, $\|\bm{k}(\bm{x}_{s,k}(t;\bm{x}_i))-\bm{k}(\hat{\bm{x}}_i)\|\leq 2\overline{u}$ holds.
  Thus, we have the conclusion.
  \end{proof}

 Over the time interval $[t_i,t_{i+1})$, system \eqref{sc_system_k} is equivalent to system \eqref{p_system} subject to the perturbation input $\bm{d}(t)=-\bm{k}(\bm{x}_{s,k}(t-t_i;\bm{x}_j))+\bm{k}(\widehat{\bm{x}}_i)$. Therefore, according to Theorem \ref{perturbed_reach_avoid} and Lemma \ref{perturbation_bound}, if $\overline{d}$ in system \eqref{p_system} is equal to $\min \{\beta \alpha \Delta+\beta \epsilon,2\overline{u}\}$ and the feedback controller $\bm{k}(\bm{x}): \mathbb{R}^n \rightarrow \mathcal{U}$ satisfies \eqref{enhanced}, i.e.,
   \begin{equation}
   \label{enhanced0}
   \begin{split}
   \mathcal{L}_{\bm{f}}&h_{\mathcal{C}}(\bm{x})+\mathcal{L}_{\bm{g}}h_{\mathcal{C}}(\bm{x})\bm{k}(\bm{x})\\
   &-\lambda h_{\mathcal{C}}(\bm{x})\geq \gamma \min \{\beta \alpha \Delta+\beta \epsilon,2\overline{u}\}, \forall \bm{x}\in \mathcal{C}\setminus \mathcal{G},
   \end{split}
    \end{equation}
then the sampled-data system \eqref{sc_system_k} satisfies the reach-avoid property with respect to the safe set $\mathcal{C}$ and goal set $\mathcal{G}$. 

 \begin{lemma}
  \label{lemma:sampled}
   If the controller $\bm{k}(\bm{x}): \mathbb{R}^n \rightarrow \mathcal{U}$ satisfies 
   \eqref{enhanced0}, then the sampled-data system \eqref{sc_system_k} satisfies the reach-avoid property with respect to the safe set $\mathcal{C}$ and goal set $\mathcal{G}$. 
 \end{lemma}
 \begin{proof}
     Consider system \eqref{p_system} with $\overline{d}=\max\{\beta \alpha \Delta+\beta \epsilon,2\overline{u}\}$. Theorem \ref{perturbed_reach_avoid} ensures that system \eqref{p_system} satisfies the reach-avoid property with respect to the safe set $\mathcal{C}$ and goal set $\mathcal{G}$, regardless of any perturbation input $\|\bm{d}\|\leq \overline{d}$.

    Given $\bm{x}_0\in \mathcal{C}$, let $\tau$ be the maximal time instant such that system \eqref{sc_system_k} stay inside the set $\mathcal{C}\setminus \mathcal{G}$ over the time interval $[0,\tau)$. In the following we first show that $\tau<\infty$. 

Since $\bm{k}(\hat{\bm{x}}_j)=\bm{k}(\hat{\bm{x}}_j)-\bm{k}(\bm{x}_{s,k}(t-t_j;\bm{x}_j))+\bm{k}(\bm{x}_{s,k}(t-t_j;\bm{x}_j))$ for $t\in [t_j,t_{j+1})$ with $j\in \mathbb{N}_{\geq 0}$, and \[\|\bm{k}(\hat{\bm{x}}_j)-\bm{k}(\bm{x}_{s,k}(t-t_j;\bm{x}_j))\|\leq \max\{\beta \epsilon+\beta \alpha \Delta,2\overline{u}\}\] for $j\in \mathbb{N}_{\geq 0}$, where $\hat{\bm{x}}_j$ is the measured state of system \eqref{sc_system_k} at time $t_j$, according to Lemma \ref{perturbation_bound} we conclude that system \eqref{sc_system_k} is equivalent to system \eqref{p_system} with 
\begin{equation}
\label{per_sa}
\bm{d}(t)=\bm{k}(\hat{\bm{x}}_j)-\bm{k}(\bm{x}_{s,k}(t-t_j;\bm{x}_j))
\end{equation}
for $t\in [t_j,t_{j+1}]$ with $j\in \mathbb{N}_{\geq 0}$.
Thus, we have that 
system \eqref{p_system} with the perturbation input \eqref{per_sa} will leave the set $\mathcal{C}\setminus \mathcal{G}$ and enter the goal set $\mathcal{G}$ in finite time. Therefore, $\tau<\infty$. 
Thus, the sampled-data system \eqref{sc_system_k} with $\bm{x}_0\in \mathcal{C}$ will reach the goal set $\mathcal{G}$ at $t=\tau$ while staying inside the safe set $\mathcal{C}$ before $\tau$. 
 \end{proof}

The observability of the state of system \eqref{sc_system_k} is limited to sampling time instants, meaning that the confirmation of entering the goal set $\mathcal{X}$ can only be obtained through measured states at these specific time points in practice. In the subsequent subsection, we aim to improve constraint \eqref{enhanced0} by refining the goal set, thus ensuring goal attainment through the use of measured states.

\subsection{Sufficient Conditions of Reach-avoid Satisfaction for System \eqref{sc_system_k}}
\label{sec:scrass}
In this subsection we will present our sufficient condition for solving Problem \ref{reach_avoid_set_sampled}. 

A Lipschitz controller, denoted as $k(\bm{x})$, satisfying condition \eqref{enhanced0}, ensures that system \eqref{sc_system_k} will eventually enter the goal set $\mathcal{G}$. However, it is possible for the system to enter and leave the goal set between sampling time instances, potentially causing the event of reaching the goal set to be missed. To address this practical concern, we will define a subset $\widehat{\mathcal{G}}$ of the goal set $\mathcal{G}$. If the measured state at a sampling instance (i.e., $t_i=i\Delta$) of the sampled-data system \eqref{sc_system_k} falls within $\widehat{\mathcal{G}}$, system \eqref{sc_system_k} will definitely enter the goal set $\mathcal{G}$. To better understand the determination of the set $\widehat{\mathcal{G}}$, we can break it down into three steps:
\begin{enumerate}
    \item Find a subset $\mathcal{G}_1$ of the goal set $\mathcal{G}$. If the measured state falls within $\mathcal{G}_1$, then the exact state is in the goal set $\mathcal{G}$.
    \item Determine a subset $\mathcal{G}_2$ of the set $\mathcal{G}$. If the exact state is in $\mathcal{G}_2$, then the measured state belongs to $\mathcal{G}_1$.
    \item Determine a subset $\widehat{\mathcal{G}}$ of the goal set $\mathcal{G}$. If there exists a time $\tau$ in the interval $[t_i, t_{i+1})$ such that $\bm{x}_{s,k}(\tau-t_i;\bm{x}_i) \in \widehat{\mathcal{G}}$, $\bm{x}_i \in \mathcal{G}_2$ holds.
\end{enumerate} 

The subsets $\mathcal{G}_1$, $\mathcal{G}_2$ and $\widehat{\mathcal{G}}$ are respectively formulated in Lemma \ref{G_1}, \ref{G_2} and \ref{reach_target_integer}.  

\begin{lemma}
\label{G_1}
    Let $\hat{\bm{x}}$ be the measured state of the state $\bm{x}$ with $\|\hat{\bm{x}}-\bm{x}\|\leq \epsilon$. Then if $\hat{\bm{x}} $ falls within \[\mathcal{G}_1=\{\bm{y}\in \mathcal{G}\mid h_{\mathcal{G}}(\hat{\bm{y}}) > \xi\epsilon\},\] $\bm{x}\in \mathcal{G}$ holds. 
\end{lemma}
\begin{proof}
    Since \[|h_{\mathcal{G}}(\hat{\bm{x}})-h_{\mathcal{G}}(\bm{x})|\leq \xi \|\hat{\bm{x}}-\bm{x}\|\leq \xi\epsilon,\] $h(\bm{x})>0$ holds if  $h_{\mathcal{G}}(\hat{\bm{x}}) > \xi \epsilon$. Thus, we have the conclusion. 
\end{proof}

\begin{lemma}
\label{G_2}
    Let $\hat{\bm{x}}$ be the measured value of the state $\bm{x}$ with $\|\hat{\bm{x}}-\bm{x}\|\leq \epsilon$. Then if $\bm{x}$ falls within  \[\mathcal{G}_2=\{\bm{y}\in \mathcal{G}\mid h_{\mathcal{G}}(\hat{\bm{y}}) > 2\xi\epsilon\},\] $\hat{\bm{x}}\in \mathcal{G}_1$ holds.  
\end{lemma}
\begin{proof}
    Since $|h_{\mathcal{G}}(\hat{\bm{x}})-h_{\mathcal{G}}(\bm{x})|\leq \xi \|\hat{\bm{x}}-\bm{x}\|\leq \xi\epsilon$, we have that $h(\hat{\bm{x}})>\xi \epsilon$ holds if  $h_{\mathcal{G}}(\bm{x}) > 2\xi \epsilon$. Thus, the conclusion holds.
\end{proof}

\begin{lemma}
\label{reach_target_integer}
If $\tau \in [t_i,t_{i+1}]$ with $i\in \mathbb{N}$ is a time instant satisfying \[\bm{x}_{s,k}(\tau-t_i;\bm{x}_i)\in \widehat{\mathcal{G}}=\{\bm{y}\in \mathcal{G}\mid h(\bm{y})>\xi\alpha \Delta+2\xi\epsilon\},\] then $\bm{x}_i\in \mathcal{G}_2$ holds.  
\end{lemma}
\begin{proof}
    Since $|h_{\mathcal{G}}(\bm{x}_{s,k}(\tau-t_i;\bm{x}_i))-h_{\mathcal{G}}(\bm{x}_i)|\leq \xi \|\bm{x}_{s,k}(\tau-t_i;\bm{x}_i)-\bm{x}_i\|$ and $\|\bm{x}_{s,k}(\tau-t_i;\bm{x}_i)-\bm{x}_i\|\leq \alpha (\tau-t_i)\leq \alpha \Delta$, we have \[|h_{\mathcal{G}}(\bm{x}_{s,k}(\tau-t_i;\bm{x}_i))-h_{\mathcal{G}}(\bm{x}_i)|\leq \xi \alpha \Delta.\] Also, since $h_{\mathcal{G}}(\bm{x}_{s,k}(\tau-t_i;\bm{x}_i))>\xi\alpha \Delta+2\xi\epsilon$, we have the conclusion. 
\end{proof}


Now, we have our solution to Problem \ref{reach_avoid_set_sampled}, which is formally stated in Theorem \ref{thm:sampled} . 
  \begin{theorem}
  \label{thm:sampled}
   If the controller $\bm{k}(\bm{x}): \mathbb{R}^n \rightarrow \mathcal{U}$ to system \eqref{system} satisfies   
   \begin{equation}
   \label{enhanced}
   \begin{split}
   \mathcal{L}_{\bm{f}}&h_{\mathcal{C}}(\bm{x})+\mathcal{L}_{\bm{g}}h_{\mathcal{C}}(\bm{x})\bm{k}(\bm{x})\\
   &-\lambda h_{\mathcal{C}}(\bm{x})\geq \gamma \min \{\beta \alpha \Delta+\beta \epsilon,2\overline{u}\}, \forall \bm{x}\in \mathcal{C}\setminus \widehat{\mathcal{G}},
   \end{split}
    \end{equation}
 where \[\widehat{\mathcal{G}}=\{\bm{x}\in \mathcal{C}\mid  h_{\mathcal{G}}(\bm{x})-\xi \alpha \Delta-2\xi\epsilon>0\} \neq \emptyset,\] then the sampled-data system \eqref{sc_system_k} satisfies the reach-avoid property with respect to the safe set $\mathcal{C}$ and goal set $\mathcal{G}$. Specially, there exists a sampling time instant $t_i$ such that the sampled-data system \eqref{sc_system_k} is confirmed to enter the goal set $\mathcal{G}$ via its measured state $\widehat{\bm{x}}_i$ at $t=t_i$, and it does not leave the safe set $\mathcal{C}$ before $t_i$.  
 \end{theorem}
 \begin{proof}
Since $\widehat{\mathcal{G}}\subseteq \mathcal{G}$, the conclusion that system \eqref{sc_system_k} with $\bm{x}(0)=\bm{x}_0\in \mathcal{C}$ satisfies the reach-avoid property with respect to the safe set $\mathcal{C}$ and goal set $\mathcal{G}$ can be confirmed by Lemma \ref{lemma:sampled}.

In addition, assume that $\tau \in [t_i,t_{i+1})$ is the first target hitting time of the set $\widehat{\mathcal{G}}$.  According to Lemma \ref{reach_target_integer}, we conclude that the sampled-data system \eqref{sc_system_k} with $\bm{x}_0\in \widehat{\mathcal{C}}$ will reach the set $\mathcal{G}_2$ at $t=t_i$ while staying inside the constraint set $\mathcal{C}$ before $t_i$. Therefore, its measured state $\widehat{\bm{x}}_i$ at $t=t_i$ falls within $\mathcal{G}_1$ from Lemma \ref{G_2}. This measurement assures that the sampled-data system \eqref{sc_system_k} enters the goal set $\mathcal{G}$ according to Lemma \ref{G_1}. 
 \end{proof}

 \begin{remark}
Although we can ensure safety of the system \eqref{sc_system_k} with $\bm{x}(0)=\bm{x}_0\in \mathcal{C}$ before entering the goal set $\mathcal{G}$, we cannot ensure that the measured state will always stay inside the safe set $\mathcal{C}$ before hitting the goal set. 
 \end{remark}

It is worth to note that in the sampled-data system \eqref{sc_system_k}, a high sampling frequency (small sampling period $\Delta$) and small state measurement error result in minimal modifications to constraint \eqref{egbfc} to ensure that the sampled-data implementation of a Lipschitz controller safely drives the system towards the goal set. Specifically, by setting $\Delta=0$ and $\epsilon=0$, constraint \eqref{enhanced} essentially becomes equivalent to \eqref{egbfc}. To explain this effect from a set perspective, let $\hat{\mathcal{C}}=\{\bm{x}\in \mathbb{R}^n \mid h_{\mathcal{C}}(\bm{x})>-\frac{\gamma \min \{\beta \alpha \Delta+\beta \epsilon,2\overline{u}\}}{\lambda}\}\subseteq \mathcal{D}$, and suppose the Lipschitz controller $k(\bm{x})$ satisfies 
\begin{equation*}
      \mathcal{L}_{\bm{f}}h_{\mathcal{C}}(\bm{x})+\mathcal{L}_{\bm{g}}h_{\mathcal{C}}(\bm{x})\bm{k}(\bm{x})-\lambda h_{\mathcal{C}}(\bm{x})\geq 0, \forall \bm{x}\in \mathcal{D}\setminus \mathcal{G}.
  \end{equation*}
In this case, we can ensure that system \eqref{system} with the controller $\bm{f}(\bm{x})$ satisfies the reach-avoid property with respect to the safe set $\mathcal{C}$ and goal set $\mathcal{G}$. However, if the Lipschitz controller $k(\bm{x})$ satisfies 
\begin{equation}
\label{expand}
\begin{split}
      \mathcal{L}_{\bm{f}}&h_{\mathcal{C}}(\bm{x})+\mathcal{L}_{\bm{g}}h_{\mathcal{C}}(\bm{x})\bm{k}(\bm{x})\\
      &-\lambda h_{\mathcal{C}}(\bm{x})\geq  \gamma \min \{\beta \alpha \Delta+\beta \epsilon,2\overline{u}\}, \forall \bm{x}\in \mathcal{D}\setminus \hat{\mathcal{G}},
\end{split}
  \end{equation}
then we can ensure that system \eqref{system} with the controller $\bm{k}(\bm{x})$ will satisfy the reach-avoid property with respect to the expanded safe set $\hat{\mathcal{C}}$ and the shrunk goal set $\hat{\mathcal{G}}$.

\section{Examples}
\label{sec:expe}
In this section, we demonstrate our theoretical developments on two systems, i.e.,  an inverted pendulum system and a cruise control system.

\begin{example}
In this example, we have an inverted pendulum system with the state $\bm{x} = [\theta, \dot{\theta}]^\top$, where $\theta$ represents the pendulum angle and $\dot{\theta}$ represents the angular velocity. The dynamics of the system are described by the following differential equation \cite{alan2023control}:
    \begin{equation*}
        \frac{d}{dt} \left[ \begin{array}{c}
           \theta  \\
           \dot{\theta}\\ 
        \end{array} \right] = \left[\begin{array}{c}
            \dot{\theta}\\
            \frac{g}{l}\sin{\theta}\\ 
        \end{array}\right] + \left[\begin{array}{c}
            0\\
            \frac{1}{ml^2}
        \end{array} \right]u,
    \end{equation*}
    where $m$, $l$ and $g$ represent the pendulum mass, length of the rod, and gravitational acceleration, respectively. The values for  $m$, $l$ and $g$ are respectively $m = 0.04\ [kg]$, $l = 5\ [m]$ and $g=10\ [m/s^2]$. In addition, the associated sets are $\mathcal{D} = \{\bm{x} \in \mathbb{R}^2 \mid 1.1 - 4 \theta^2 - 2\dot{\theta}^2 > 0\}$, $\mathcal{C} = \{\bm{x} \in \mathcal{D} \mid 1 - 4 \theta^2 - 2\dot{\theta}^2 > 0\}$, and $\mathcal{G} = \{\bm{x} \in \mathcal{C} \mid 0.05 - \dot{\theta}^2 > 0\}$ and $\mathcal{U}=[-2.6, 2.6]$.

   These configurations $\Delta = 0.003$, $\epsilon = 0.001$, $\lambda = 0.001$ are used in condition \eqref{thm:sampled}. We can ensure that the following controller \eqref{u2} satisfies constraint \eqref{enhanced}.

    \begin{lemma}
        The controller
    \begin{equation}
    \label{u2}
        k([\theta, \dot{\theta}]^\top) = -2\theta-2\sin{\theta} - 2\Dot{\theta}
    \end{equation}
    satisfies condition \eqref{enhanced}.
    \end{lemma}

    \begin{proof}
   Upon calculations, we can obtain $\alpha=3.6014$, $\beta = 4.4721$, $\gamma = 2.9665$, $\xi = 1.4832$, \[\gamma \min \{\beta \alpha \Delta+\beta \epsilon,2\overline{u}\} < 0.1566,\] \[\{\bm{x}\in\mathcal{C}\mid 0.0311 - \dot{\theta}^2>0\} \subseteq \widehat{\mathcal{G}} \subseteq \{\bm{x}\in\mathcal{C}\mid 0.0310 - \dot{\theta}^2>0\}\] and $\lambda h_c(\bm{x}) < \lambda$. Therefore, we have
    \begin{alignat*}{2}
        &\mathcal{L}_{\bm{f}}h_{\mathcal{C}}(\bm{x})+\mathcal{L}_{\bm{g}}h_{\mathcal{C}}(\bm{x})\bm{k}(\bm{x})-\lambda h_{\mathcal{C}}(\bm{x}) \\
        =&-8\theta\dot{\theta} - 8\dot{\theta}\sin{\theta} + 8\theta\dot{\theta} + 8\dot{\theta}\sin{\theta} + 8\dot{\theta}^2 -\lambda h_{\mathcal{C}}(\bm{x})\\
        =&8\dot{\theta}^2-\lambda h_{\mathcal{C}}(\bm{x}) \\
        >&0.2481 - 0.001 \\
        >&\gamma \min \{\beta \alpha \Delta+\beta \epsilon,2\overline{u}\}, \forall \bm{x}\in \mathcal{C}\setminus \widehat{\mathcal{G}}.
    \end{alignat*}
 Consequently, controller \eqref{u2} satisfies constraint \eqref{enhanced}.
    \end{proof}
    
    

The trajectory of system \eqref{sc_system_k} with the controller \eqref{u2}, subject to a sampling time interval of $\Delta = 0.003$ and a measurement perturbation of $\epsilon = 0.001$, is plotted in Figure \ref{fig:exp_pendulum_contrast} in blue. The initial condition of the system is set to $\bm{x}_0=[-0.4,0.3]^{\top}$. To compute the trajectory between the sampling instants, Runge-Kutta methods are used. The states at the sampling time instants are obtained by perturbing the computed states with a random perturbation of magnitude $\epsilon = 0.001$. Some measured states can be seen in the embedded plot located on the left side of Figure \ref{fig:exp_pendulum_contrast}, which shows a zoomed version of the trajectory close to the set $\mathcal{G}_1$ (according to Lemma \ref{G_1}, if the measured state falls within $\mathcal{G}_1$, system \eqref{sc_system_k} enters the target set $\mathcal{G}$). From the plot, it can be observed that system \eqref{sc_system_k} is able to enter the target set $\mathcal{G}$ safely, and the goal reach can be ensured by using the measured state. The measured state falling with the set $\mathcal{G}_1$ is represented by the magenta asterisk, and the corresponding actual exact state is denoted by the blue asterisk which lies within the goal set $\mathcal{G}$. 

  Figure \ref{fig:exp_pendulum_contrast} also shows another scenario of system \eqref{sc_system_k} with the controller \eqref{u2}. This scenario is generated with a sampling time interval of $\Delta = 0.0003$ and a measurement perturbation of $\epsilon = 0.001$, which also satisfies condition \eqref{enhanced}. The resulting trajectory (yellow curve in Fig. \ref{fig:exp_pendulum_contrast}) aligns with the previous scenario and is indistinguishable from it. The embedded plot on the right side of Figure \ref{fig:exp_pendulum_contrast} zooms in on a portion of the trajectory and the measured states near the set $\mathcal{G}_1$. By comparing these two scenarios, corresponding to $(\Delta, \epsilon) = (0.003,0.001)$ and $(\Delta, \epsilon) = (0.0003,0.001)$, we can conclude that a smaller $\Delta$ results in a smaller deviation of the trajectory from the one induced by the continuous-time controller \eqref{u2}, although the associated measured states may appear more messy. This is because the sampled-data controller resulting from the smaller $\Delta$ better matches the continuous-time controller \eqref{u2}, as visualized in Figure \ref{fig:exp_pendulum_k_contrast}.



    \begin{figure}[htb!]
    \center
     \includegraphics[width=1.0\linewidth,height=0.7\linewidth]{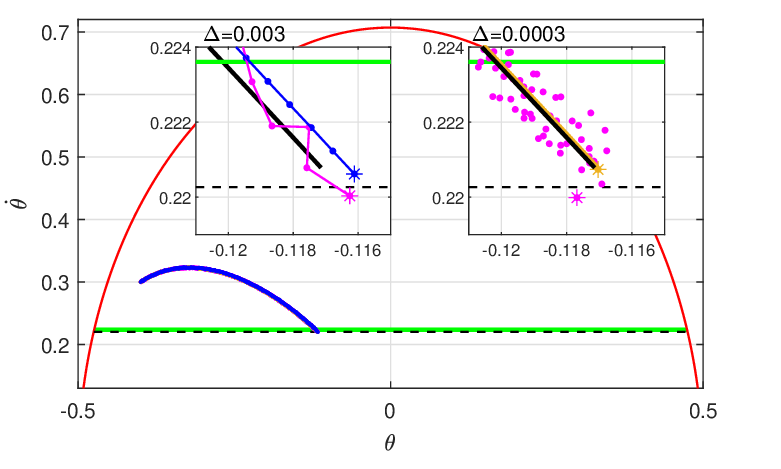}
    \caption{Red and green curve  denotes a part of the boundary $\partial \mathcal{C}$ and $\partial \mathcal{G}$, respectively; dashed black curve denotes a part of the boundary $\partial \mathcal{G}_1$; blue and yellow curve denote the trajectory of system \eqref{sc_system_k} with $(\Delta,\epsilon)=(0.003,0.001)$ and $(\Delta,\epsilon)=(0.0003,0.001)$ starting from the initial state $[-0.4,0.3]^{\top}$, respectively; black curve represents the trajectory of the system without measurement uncertainties under the continuous-time controller \eqref{u2}; magenta points denote measured states. }
    \label{fig:exp_pendulum_contrast}
    \end{figure}

    \begin{figure}[htb!]
    \center
     \includegraphics[width=\linewidth,height=0.9\linewidth]{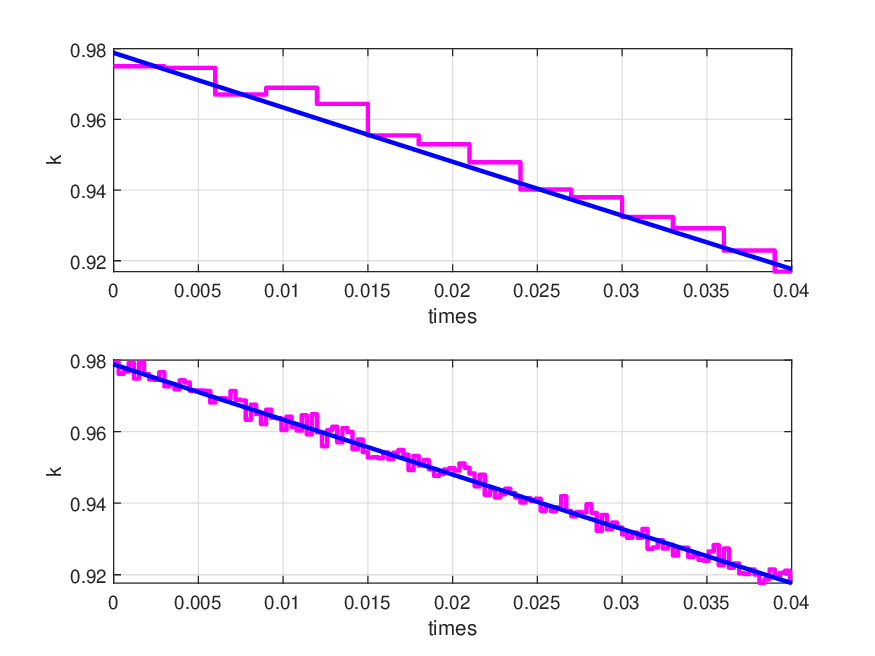}
    \caption{The two figures illustrate the controller \eqref{u2} and its sampled-data implementation for the system starting from the initial state $[-0.4;0.3]$ over the time horizon $[0,0.04]$. The top figure corresponds to the case where $(\Delta,\epsilon)=(0.003,0.001)$, while the bottom figure corresponds to the case where $(\Delta,\epsilon)=(0.0003,0.001)$. Magenta curve represents sampled-data control and blue curve represents the continuous-time controller.}
    \label{fig:exp_pendulum_k_contrast}
    \end{figure}


    
\end{example}

\begin{example}
\label{exp_acc}
Consider a cruise control system describing automatic car-following, in which the primary objective is to decrease the speed of the following vehicle while simultaneously maintaining a certain distance from the leading vehicle.  \oomit{Since the control imposed on the vehicle is discrete, the sensors on the vehicle cannot continuously measure data, and the measured data has errors, the system is actually a sampled-data system with measurement uncertainties.} The dynamics of the system are described as follows \cite{alan2021safe}
\begin{equation*}
    \dot{D} = v_2- v_1, \dot{v}_1 = u, \dot{v}_2 = -v_2,
\end{equation*}
with $\mathcal{D} = \{\bm{x} \in \mathbb{R}^3 \mid 26 - (D - 7)^2 - v_1^2 - v_2^2 > 0\}$, the constraint set $\mathcal{C} = \{\bm{x} \in \mathcal{D} \mid 25 - (D - 7)^2 - v_1^2 - v_2^2 > 0\}$, the goal set $\mathcal{G} = \{\bm{x} \in \mathcal{C} \mid 4 - v_1^2 > 0\}$ and $u\in [-10, 15]$, where $D$, $v_1$, and $v_2$ represent the relative distance between the two vehicles, the speed of the following vehicle and the speed of the leading vehicle, respectively. 

These configurations $\Delta = 0.002$, $\epsilon = 0.01$ and $\lambda = 0.01$ are used in condition \eqref{thm:sampled}. A controller satisfying constraint \eqref{enhanced} is shown below. 

\begin{lemma}
The controller
\begin{equation}  
\label{u1}
    k(\bm{x})=
    \begin{cases}
        &\frac{7v_2-Dv_2+v_2^2}{v_1}+D - v_1 -7,\ \textbf{if}\ v_1^2\geq3\\
        &\frac{7v_2-Dv_2+v_2^2}{\sqrt{3}} +D - \sqrt{3}-7,\ \textbf{if}\ v_1^2<3\\
    \end{cases}
\end{equation}
satisfies condition \eqref{enhanced}.
\end{lemma}
\begin{proof}
Upon calculation, we obtain $\alpha=17.1193$, $\beta = 12.0717$, $\gamma = 10.1980$, $\xi = 10.1980$. Thus, we can conclude that \[\gamma \min \{\beta \alpha \Delta+\beta \epsilon,2\overline{u}\} < 5.4461\] and \[\{\bm{x}\in\mathcal{C}\mid 3.45 - v_1^2>0\}\subseteq \widehat{\mathcal{G}} \subseteq \{\bm{x}\in\mathcal{C}\mid 3.44 - v_1^2>0\}\] and $\lambda h_c(\bm{x}) \leq 0.25$. Therefore, we have
\begin{equation*}
\begin{split}
    &\mathcal{L}_{\bm{f}}h_{\mathcal{C}}(\bm{x})+\mathcal{L}_{\bm{g}}h_{\mathcal{C}}(\bm{x})\bm{k}(\bm{x})-\lambda h_{\mathcal{C}}(\bm{x}) \\
    =&-2(D-7)(v_2-v_1) - 2v_1u+2 v_2^2- \lambda h_{\mathcal{C}}(\bm{x})\\
    =&2v_1^2 -\lambda h_{\mathcal{C}}(\bm{x})\\
    >&6.8937 - 0.25 \\
    >&\gamma \min \{\beta \alpha \Delta+\beta \epsilon,2\overline{u}\}, \forall \bm{x}\in \mathcal{C}\setminus \widehat{\mathcal{G}}.
\end{split}
\end{equation*}
According to Theorem \ref{thm:sampled}, we have that the controller \eqref{u1} satisfies constraint \eqref{enhanced}.
\end{proof}

Figure \ref{fig:exp_follow_3D} showcases four trajectories of system \eqref{sc_system_k}, along with their corresponding initial states presented in the top left corner. The top right corner of Figure \ref{fig:exp_follow_3D} shows a zoomed version of one trajectory, together with the measured states at sampling time instants. This zoomed plot focuses on the trajectory close to the set $\mathcal{G}_1$. It can be observed that a measured state, represented by the magenta asterisk, successfully falls within the set $\mathcal{G}_1$, of which the boundary is represented by the green region. According to Lemma \ref{G_1}, this ensures that system \eqref{sc_system_k} enters the goal set $\mathcal{G}$, of which the boundary is represented by the blue region. In fact, the corresponding actual exact state, denoted by the blue asterisk, indeed falls within the goal set $\mathcal{G}$. The sampled-data controller for this trajectory is visualized in Fig. \ref{fig:exp_follow_k_new}. A zoomed portion is shown in the right top corner.




\begin{figure}[htb!]
    \center
     \includegraphics[width=\linewidth]{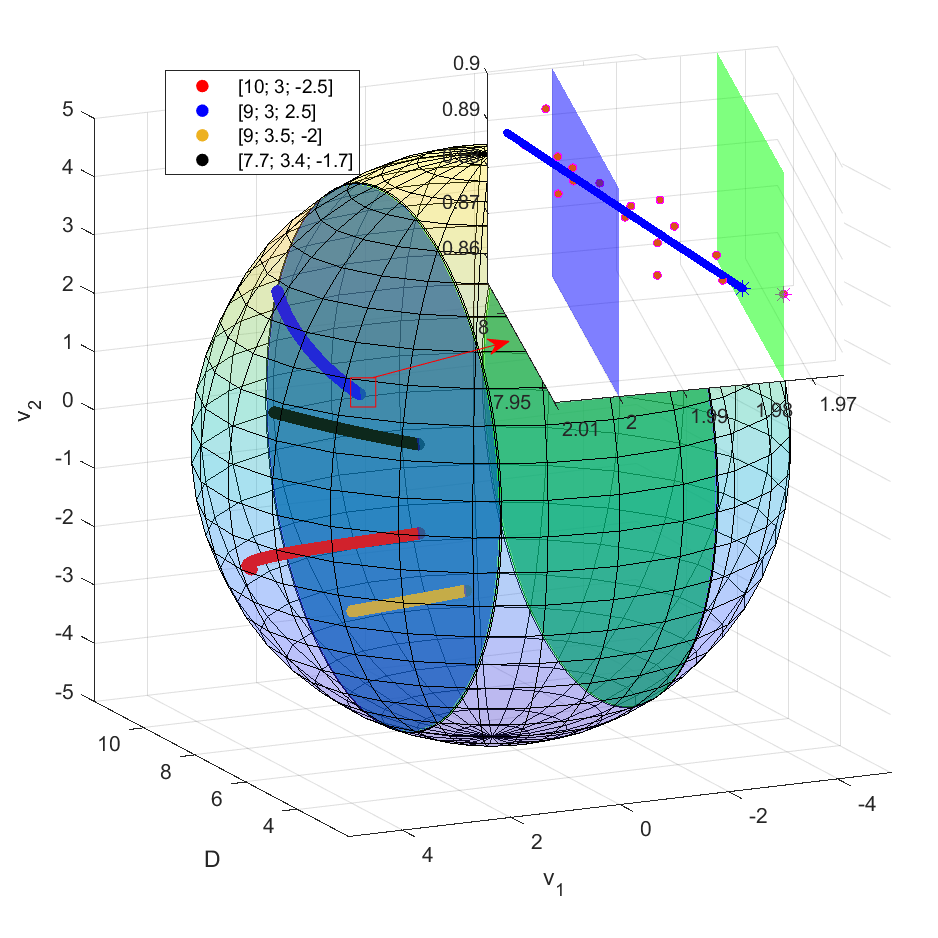}
    \caption{The sphere, blue plane, and green plane represent the boundaries of $\partial{\mathcal{C}}$, $\partial{\mathcal{G}}$, and $\partial{\mathcal{{G}}_1}$, respectively. In the zoomed plot in the right top corner, blue and magenta dots represent exact and measured states, respectively.}
    \label{fig:exp_follow_3D}
\end{figure}


\begin{figure}[htb!]
    \center
     \includegraphics[width=\linewidth,height=0.6\linewidth]{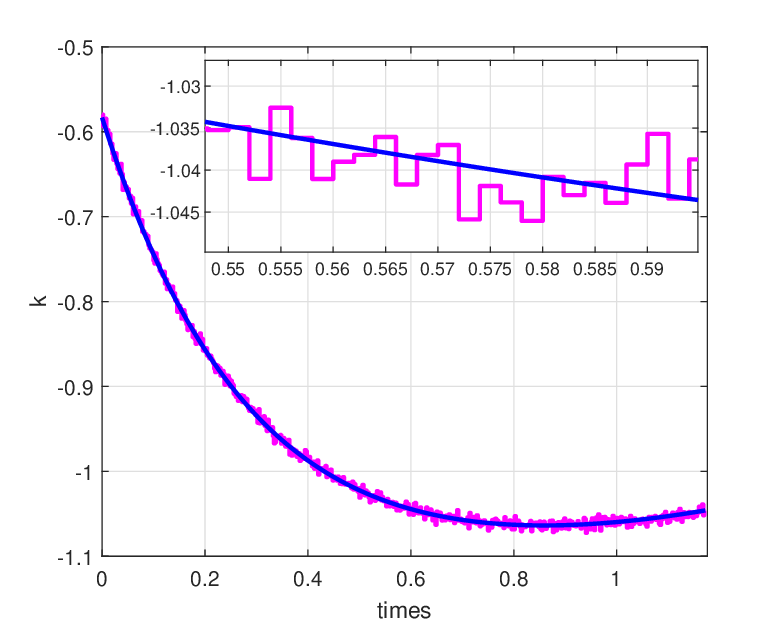}
    \caption{The blue curve represents $\bm{k}(\bm{x})$ and  magenta curve represents $\bm{k}(\widehat{\bm{x}}_i)$, which is the sampled-data controller in system \eqref{sc_system_k}}
    \label{fig:exp_follow_k_new}
\end{figure}



\end{example}

\section{Conclusion}
This paper explored the reach-avoid problem for sampled-data control systems. Specifically, we aim to determine a sufficient condition for the safe guidance of a system into a goal set and the confirmation of goal reach based on measured states at sampling time instants.  By analyzing the dynamic discrepancies between a locally Lipschitz continuous-time controller and its sampled-data implementation, we established a sufficient condition using exponential control guidance-barrier functions. The proposed method was demonstrated through two examples, showcasing the practical application of the theoretical developments. 


\bibliographystyle{abbrv}
\bibliography{ref}

\begin{thebibliography}{10}

\bibitem{ackermann2012sampled}
J.~Ackermann.
\newblock {\em Sampled-data control systems: analysis and synthesis, robust
  system design}.
\newblock Springer Science \& Business Media, 2012.

\bibitem{alan2023control}
A.~Alan, A.~J. Taylor, C.~R. He, A.~D. Ames, and G.~Orosz.
\newblock Control barrier functions and input-to-state safety with application
  to automated vehicles.
\newblock {\em IEEE Transactions on Control Systems Technology}, 2023.

\bibitem{alan2021safe}
A.~Alan, A.~J. Taylor, C.~R. He, G.~Orosz, and A.~D. Ames.
\newblock Safe controller synthesis with tunable input-to-state safe control
  barrier functions.
\newblock {\em IEEE Control Systems Letters}, 6:908--913, 2021.

\bibitem{ames2016control}
A.~D. Ames, X.~Xu, J.~W. Grizzle, and P.~Tabuada.
\newblock Control barrier function based quadratic programs for safety critical
  systems.
\newblock {\em IEEE Transactions on Automatic Control}, 62(8):3861--3876, 2016.

\bibitem{breeden2021control}
J.~Breeden, K.~Garg, and D.~Panagou.
\newblock Control barrier functions in sampled-data systems.
\newblock {\em IEEE Control Systems Letters}, 6:367--372, 2021.

\bibitem{buso2015digital}
S.~Buso and P.~Mattavelli.
\newblock {\em Digital control in power electronics}.
\newblock Morgan \& Claypool Publishers, 2015.

\bibitem{chen2012optimal}
T.~Chen and B.~A. Francis.
\newblock {\em Optimal sampled-data control systems}.
\newblock Springer Science \& Business Media, 2012.

\bibitem{cortez2019control}
W.~S. Cortez, D.~Oetomo, C.~Manzie, and P.~Choong.
\newblock Control barrier functions for mechanical systems: Theory and
  application to robotic grasping.
\newblock {\em IEEE Transactions on Control Systems Technology},
  29(2):530--545, 2019.

\bibitem{fan2018controller}
C.~Fan, U.~Mathur, S.~Mitra, and M.~Viswanathan.
\newblock Controller synthesis made real: Reach-avoid specifications and linear
  dynamics.
\newblock In {\em International Conference on Computer Aided Verification},
  pages 347--366. Springer, 2018.

\bibitem{fisac2015reach}
J.~F. Fisac, M.~Chen, C.~J. Tomlin, and S.~S. Sastry.
\newblock Reach-avoid problems with time-varying dynamics, targets and
  constraints.
\newblock In {\em Proceedings of the 18th international conference on hybrid
  systems: computation and control}, pages 11--20, 2015.

\bibitem{gurriet2019realizable}
T.~Gurriet, P.~Nilsson, A.~Singletary, and A.~D. Ames.
\newblock Realizable set invariance conditions for cyber-physical systems.
\newblock In {\em 2019 American Control Conference (ACC)}, pages 3642--3649.
  IEEE, 2019.

\bibitem{hetel2017recent}
L.~Hetel, C.~Fiter, H.~Omran, A.~Seuret, E.~Fridman, J.-P. Richard, and S.~I.
  Niculescu.
\newblock Recent developments on the stability of systems with aperiodic
  sampling: An overview.
\newblock {\em Automatica}, 76:309--335, 2017.

\bibitem{majumdar2014convex}
A.~Majumdar, R.~Vasudevan, M.~M. Tobenkin, and R.~Tedrake.
\newblock Convex optimization of nonlinear feedback controllers via occupation
  measures.
\newblock {\em The International Journal of Robotics Research},
  33(9):1209--1230, 2014.

\bibitem{margellos2011hamilton}
K.~Margellos and J.~Lygeros.
\newblock Hamilton--jacobi formulation for reach--avoid differential games.
\newblock {\em IEEE Transactions on automatic control}, 56(8):1849--1861, 2011.

\bibitem{niu2021safety}
L.~Niu, H.~Zhang, and A.~Clark.
\newblock Safety-critical control synthesis for unknown sampled-data systems
  via control barrier functions.
\newblock In {\em 2021 60th IEEE Conference on Decision and Control (CDC)},
  pages 6806--6813. IEEE, 2021.

\bibitem{singletary2020control}
A.~Singletary, Y.~Chen, and A.~D. Ames.
\newblock Control barrier functions for sampled-data systems with input delays.
\newblock In {\em 2020 59th IEEE Conference on Decision and Control (CDC)},
  pages 804--809. IEEE, 2020.

\bibitem{verdier2022formal}
C.~F. Verdier, N.~Kochdumper, M.~Althoff, and M.~Mazo~Jr.
\newblock Formal synthesis of closed-form sampled-data controllers for
  nonlinear continuous-time systems under stl specifications.
\newblock {\em Automatica}, 139:110184, 2022.

\bibitem{xue2023reach}
B.~Xue.
\newblock Reach-avoid controllers synthesis for safety critical systems.
\newblock {\em arXiv preprint arXiv:2302.14565}, 2023.

\bibitem{xue2023reach_ode}
B.~Xue, N.~Zhan, M.~Fr{\"a}nzle, J.~Wang, and W.~Liu.
\newblock Reach-avoid verification based on convex optimization.
\newblock {\em IEEE Transactions on Automatic Control}, 2023.

\end{thebibliography}
\end{document}